\documentclass[12pt]{elsarticle}
\usepackage[margin=2.5cm]{geometry}
\usepackage{lipsum}

\makeatletter
\def\ps@pprintTitle{%
   \let\@oddhead\@empty
   \let\@evenhead\@empty
   \def\@oddfoot{\reset@font\hfil\thepage\hfil}
   \let\@evenfoot\@oddfoot
}
\makeatother

\usepackage{amsfonts,color,morefloats,pslatex}
\usepackage{amssymb,amsthm, amsmath,latexsym}

\newtheorem{theorem}{Theorem}
\newtheorem{lemma}[theorem]{Lemma}

\newtheorem{corollary}[theorem]{Corollary}

\newtheorem{example}[theorem]{Example}

\newcommand{\sdef}{{\mathrm{def}}}
\newcommand{\ord}{{\mathrm{ord}}}

\newcommand{\lcm}{{\mathrm{lcm}}}
\newcommand{\tr}{{\mathrm{Tr}}}

\newcommand{\gf}{{\mathrm{GF}}}

\newcommand{\support}{{\mathrm{suppt}}}

\newcommand{\PAut}{{\mathrm{PAut}}}
\newcommand{\MAut}{{\mathrm{MAut}}}
\newcommand{\GAut}{{\mathrm{Aut}}}

\newcommand{\Sym}{{\mathrm{Sym}}}

\newcommand{\wt}{{\mathtt{wt}}}

\newcommand{\Z}{\mathbb{{Z}}}

\newcommand{\m}{\mathbb{M}}

\newcommand{\cP}{{\mathcal{P}}}
\newcommand{\cB}{{\mathcal{B}}}

\newcommand{\C}{{\mathsf{C}}}
\newcommand{\M}{{\mathsf{M}}}

\newcommand{\bc}{{\mathbf{c}}}

\newcommand{\bzero}{{\mathbf{0}}}

\newcommand{\bD}{{\mathbb{D}}}

\newcommand{\PGL}{{\mathrm{PGL}}}

\newcommand{\PGaL}{{\mathrm{P \Gamma L}}}

\usepackage{blindtext}

\begin{document}

\begin{frontmatter}




\title{Infinite families of near MDS codes holding $t$-designs}

\tnotetext[fn1]{C. Ding's research was supported by the Hong Kong Research Grants Council,
Proj. No. 16300418. C. Tang was supported by The National Natural Science Foundation of China (Grant No.
11871058) and China West Normal University (14E013, CXTD2014-4 and the Meritocracy Research
Funds).}

\author[cding]{Cunsheng Ding}
\ead{cding@ust.hk}
\author[cmt]{Chunming Tang}
\ead{tangchunmingmath@163.com}

\address[cding]{Department of Computer Science and Engineering, The Hong Kong University of Science and Technology, Clear Water Bay, Kowloon, Hong Kong, China}
\address[cmt]{School of Mathematics and Information, China West Normal University, Nanchong, Sichuan,  637002, China}




\begin{abstract} 
An $[n, k, n-k+1]$ linear code is called an MDS code. 
An $[n, k, n-k]$ linear code is said to be almost maximum distance separable (almost MDS or AMDS for short). 
A code is said to be near  maximum distance separable (near MDS or NMDS for short) if the code and its dual code 
both are almost maximum distance separable. The first near MDS code was the $[11, 6, 5]$ ternary Golay code 
discovered in 1949 by Golay. This ternary code holds $4$-designs, and its extended code holds a Steiner system 
$S(5, 6, 12)$ with the largest strength known. In the past 70 years, sporadic near MDS codes 
holding $t$-designs were discovered and many infinite families of near MDS codes over finite fields were 
constructed. However, the question as to whether there is an infinite family of near MDS codes holding an infinite 
family of $t$-designs for $t\geq 2$ remains open for 70 years. This paper settles this long-standing problem by 
presenting an infinite family of near MDS codes over $\gf(3^s)$ holding an infinite family of $3$-designs and an infinite family of near MDS  codes over $\gf(2^{2s})$ holding an infinite family of $2$-designs. The subfield subcodes of these two families 
are also studied, and are shown to be dimension-optimal or distance-optimal.       
\end{abstract}


\end{frontmatter}


\section{Introduction}

Let $\cP$ be a set of $v \ge 1$ elements, and let $\cB$ be a set of $k$-subsets of $\cP$, where $k$ is
a positive integer with $1 \leq k \leq v$. Let $t$ be a positive integer with $t \leq k$. The incidence structure 
$\bD = (\cP, \cB)$ is called a $t$-$(v, k, \lambda)$ {\em design\index{design}}, or simply {\em $t$-design\index{$t$-design}}, if every $t$-subset of $\cP$ is contained in exactly $\lambda$ elements of
$\cB$. The elements of $\cP$ are called points, and those of $\cB$ are referred to as blocks.
We usually use $b$ to denote the number of blocks in $\cB$. 
Let $\binom{\cP}{k}$ denote the set of all $k$-subsets of $\cP$. Then $(\cP, \binom{\cP}{k})$ is a $k$-$(v, k, 1)$ design, 
which is called a \emph{complete design}. 
  A $t$-design is called {\em simple\index{simple}} if $\cB$ does not contain
any repeated blocks.
In this paper, we consider only simple $t$-designs with $v > k > t$.
A $t$-$(v,k,\lambda)$ design is referred to as a
{\em Steiner system\index{Steiner system}} if $t \geq 2$ and $\lambda=1$,
and is denoted by $S(t,k, v)$. 

A $t$-$(v, k, \lambda)$ design is also a $s$-$(v, k, \lambda_s)$ design with 
\begin{eqnarray}\label{eqn-lambdas}
\lambda_s=\lambda \binom{v-s}{t-s}/\binom{k-s}{t-s}
\end{eqnarray} 
for all $s$ with $0 \leq s \leq t$.

Let $\C$ be a $[v, \kappa, d]$ linear code over $\gf(q)$. Let $A_i$ denote the
number of codewords with Hamming weight $i$ in $\C$, where $0 \leq i \leq v$. The sequence
$(A_0, A_1, \cdots, A_{v})$ is
called the \textit{weight distribution} of $\C$, and $\sum_{i=0}^v A_iz^i$ is referred to as
the \textit{weight enumerator} of $\C$. In this paper, $\C^\perp$ denotes the dual code 
of $\C$, $d^\perp$ denotes the minimum distance of $\C^\perp$, and $(A_0^\perp, A_1^\perp, \cdots, A_{v}^\perp)$ 
denotes the weight distribution of $\C^\perp$.  

A  $[v, \kappa, d]$ linear code over $\gf(q)$ is said to be distance-optimal if there is no $[v, \kappa, d']$ over 
$\gf(q)$ with $d'>d$. A  $[v, \kappa, d]$ linear code over $\gf(q)$ is said to be dimension-optimal if there is no 
$[v, \kappa', d]$ over $\gf(q)$ with $\kappa' > \kappa$.  A  $[v, \kappa, d]$ linear code over $\gf(q)$ is said 
to be length-optimal if there is no $[v', \kappa, d]$ over $\gf(q)$ with $v'<v$. A linear code is said to be optimal 
if it is distance-optimal, dimension-optimal and length-optimal.  

A coding-theoretic construction of $t$-designs is the following. 
For each $k$ with $A_k \neq 0$,  let $\cB_k$ denote
the set of the supports of all codewords with Hamming weight $k$ in $\C$, where the coordinates of a codeword
are indexed by $(p_1, \ldots, p_v)$. Let $\cP=\{p_1, \ldots, p_v\}$.  The pair $(\cP, \cB_k)$
may be a $t$-$(v, k, \lambda)$ design for some positive integer $\lambda$, which is called a
\emph{support design} of the code, and is denoted by $\bD_k(\C)$. In such a case, we say that the code $\C$ holds a $t$-$(v, k, \lambda)$
design or the codewords of weight $k$ in $\C$ support a $t$-$(v, k, \lambda)$
design. 

The following theorem, developed by Assumus and Mattson, shows that the pair $(\cP, \cB_k)$ defined by 
a linear code is a $t$-design under certain conditions \cite{AM69}.

\begin{theorem}[Assmus-Mattson Theorem]\label{thm-designAMtheorem}
Let $\C$ be a $[v,k,d]$ code over $\gf(q)$. Let $d^\perp$ denote the minimum distance of $\C^\perp$. 
Let $w$ be the largest integer satisfying $w \leq v$ and 
$$ 
w-\left\lfloor  \frac{w+q-2}{q-1} \right\rfloor <d. 
$$ 
Define $w^\perp$ analogously using $d^\perp$. Let $(A_i)_{i=0}^v$ and $(A_i^\perp)_{i=0}^v$ denote 
the weight distribution of $\C$ and $\C^\perp$, respectively. Fix a positive integer $t$ with $t<d$, and 
let $s$ be the number of $i$ with $A_i^\perp \neq 0$ for $0 \leq i \leq v-t$. Suppose $s \leq d-t$. Then 
\begin{itemize}
\item the codewords of weight $i$ in $\C$ hold a $t$-design provided $A_i \neq 0$ and $d \leq i \leq w$, and 
\item the codewords of weight $i$ in $\C^\perp$ hold a $t$-design provided $A_i^\perp \neq 0$ and 
         $d^\perp \leq i \leq \min\{v-t, w^\perp\}$. 
\end{itemize}
\end{theorem}

The Assmus-Mattson Theorem is a very useful tool in constructing $t$-designs from linear codes 
(see, for example, \cite{DingLi16}, \cite{Dingbook18}, \cite{Tonchev} and \cite{Tonchevhb}). 
Another sufficient condition for the 
incidence structure $(\cP, \cB_k)$ to be a $t$-design is via the automorphism group of the 
code $\C$.

The set of coordinate permutations that map a code $\C$ to itself forms a group, which is referred to as
the \emph{permutation automorphism group\index{permutation automorphism group of codes}} of $\C$
and denoted by $\PAut(\C)$. If $\C$ is a code of length $n$, then $\PAut(\C)$ is a subgroup of the
\emph{symmetric group\index{symmetric group}} $\Sym_n$.

A \emph{monomial matrix\index{monomial matrix}} over $\gf(q)$ is a square matrix having exactly one
nonzero element of $\gf(q)$  in each row and column. A monomial matrix $M$ can be written either in
the form $DP$ or the form $PD_1$, where $D$ and $D_1$ are diagonal matrices and $P$ is a permutation
matrix.

The set of monomial matrices that map $\C$ to itself forms the group $\MAut(\C)$,  which is called the
\emph{monomial automorphism group\index{monomial automorphism group}} of $\C$. Clearly, we have
$$
\PAut(\C) \subseteq \MAut(\C).
$$

The \textit{automorphism group}\index{automorphism group} of $\C$, denoted by $\GAut(\C)$, is the set
of maps of the form $M\gamma$,
where $M$ is a monomial matrix and $\gamma$ is a field automorphism, that map $\C$ to itself. In the binary
case, $\PAut(\C)$,  $\MAut(\C)$ and $\GAut(\C)$ are the same. If $q$ is a prime, $\MAut(\C)$ and
$\GAut(\C)$ are identical. In general, we have
$$
\PAut(\C) \subseteq \MAut(\C) \subseteq \GAut(\C).
$$

By definition, every element in $\GAut(\C)$ is of the form $DP\gamma$, where $D$ is a diagonal matrix,
$P$ is a permutation matrix, and $\gamma$ is an automorphism of $\gf(q)$.
The automorphism group $\GAut(\C)$ is said to be $t$-transitive if for every pair of $t$-element ordered
sets of coordinates, there is an element $DP\gamma$ of the automorphism group $\GAut(\C)$ such that its
permutation part $P$ sends the first set to the second set. The automorphism group $\GAut(\C)$ is said to be $t$-homogeneous if for every pair of $t$-element 
sets of coordinates, there is an element $DP\gamma$ of the automorphism group $\GAut(\C)$ such that its
permutation part $P$ sends the first set to the second set.

The next theorem gives a 
sufficient condition for a linear code to hold $t$-designs \cite[p. 308]{HP03}.

\begin{theorem}\label{thm-designCodeAutm}
Let $\C$ be a linear code of length $n$ over $\gf(q)$ where $\GAut(\C)$ is $t$-transitive 
or $t$-homogeneous. Then the codewords of any weight $i \geq t$ of $\C$ hold a $t$-design.
\end{theorem} 

So far many infinite families of $3$-designs have been constructed from this coding-theoretic 
approach. However, no infinite family of $4$-designs has been produced with this approach, 
though sporadic $4$-designs and $5$-designs have been obtained from sporadic linear codes. 

An $[n, k, n-k+1]$ linear code is called an MDS code. MDS codes do hold $t$-designs with very large $t$. 
Unfortunately, all $t$-designs held in MDS codes are complete ans thus trivial. One would naturally ask if some 
$[n, k, n-k]$ codes over finite fields hold simple $t$-designs. 

An $[n, k, n-k]$ linear code is said to be almost maximum distance separable (almost MDS or AMDS for short). 
A code is said to be near  maximum distance separable (near MDS or NMDS for short) if the code and its dual code 
both are almost maximum distance separable. The first near MDS code was the $[11, 6, 5]$ ternary Golay code 
discovered in 1949 by Golay \cite{Golay49}. This ternary code holds $4$-designs, and its extended code holds a Steiner system 
$S(5, 6, 12)$ having the largest strength known. In the past 70 years, sporadic near MDS codes 
holding $t$-designs were discovered and many infinite families of near MDS codes over finite fields were 
constructed. However, the question as to whether there is an infinite family of near MDS codes holding an infinite 
family of $t$-designs for $t\geq 2$ remains open for 70 years. This paper settles this long-standing problem by 
presenting an infinite family of near MDS codes over $\gf(3^s)$ holding an infinite family of $3$-designs and an infinite family of near MDS  codes over $\gf(2^{2s})$ holding an infinite family of $2$-designs.  The subfield subcodes of the two families of 
near MDS codes are also studied, and shown to be dimension-optimal or distance-optimal.

\section{Cyclic codes and BCH codes} 

An $[n,k, d]$ code $\C$ over $\gf(q)$ is {\em cyclic} if 
$(c_0,c_1, \cdots, c_{n-1}) \in \C$ implies $(c_{n-1}, c_0, c_1, \cdots, c_{n-2}) 
\in \C$.  
By identifying any vector $(c_0,c_1, \cdots, c_{n-1}) \in \gf(q)^n$ 
with  
$$ 
c_0+c_1x+c_2x^2+ \cdots + c_{n-1}x^{n-1} \in \gf(q)[x]/(x^n-1), 
$$
any code $\C$ of length $n$ over $\gf(q)$ corresponds to a subset of the quotient ring 
$\gf(q)[x]/(x^n-1)$. 
A linear code $\C$ is cyclic if and only if the corresponding subset in $\gf(q)[x]/(x^n-1)$ 
is an ideal of the ring $\gf(q)[x]/(x^n-1)$. 

Note that every ideal of $\gf(q)[x]/(x^n-1)$ is principal. Let $\C=\langle g(x) \rangle$ be a 
cyclic code, where $g(x)$ is monic and has the smallest degree among all the 
generators of $\C$. Then $g(x)$ is unique and called the {\em generator polynomial,} 
and $h(x)=(x^n-1)/g(x)$ is referred to as the {\em parity-check} polynomial of $\C$. 

Let $n$ be a positive integer and 
let $\Z_n$ denote  the set $\{0,1,2, \cdots, n-1\}$.  Let $s$ be an integer with $0 \leq s <n$. The \emph{$q$-cyclotomic coset of $s$ modulo $n$\index{$q$-cyclotomic coset modulo $n$}} is defined by 
$$ 
C_s=\{s, sq, sq^2, \cdots, sq^{\ell_s-1}\} \bmod n \subseteq \Z_n,  
$$
where $\ell_s$ is the smallest positive integer such that $s \equiv s q^{\ell_s} \pmod{n}$, and is the size of the 
$q$-cyclotomic coset. The smallest integer in $C_s$ is called the \emph{coset leader\index{coset leader}} of $C_s$. 
Let $\Gamma_{(n,q)}$ be the set of all the coset leaders. We have then $C_s \cap C_t = \emptyset$ for any two 
distinct elements $s$ and $t$ in  $\Gamma_{(n,q)}$, and  
\begin{eqnarray}\label{eqn-cosetPP}
\bigcup_{s \in  \Gamma_{(n,q)} } C_s = \Z_n. 
\end{eqnarray}
Hence, the distinct $q$-cyclotomic cosets modulo $n$ partition $\Z_n$. 

Let $m=\ord_{n}(q)$ be the order of $q$ modulo $n$, and let $\alpha$ be a generator of $\gf(q^m)^*$. Put $\beta=\alpha^{(q^m-1)/n}$. 
Then $\beta$ is a primitive $n$-th root of unity in $\gf(q^m)$. The minimal polynomial $\m_{\beta^s}(x)$ 
of $\beta^s$ over $\gf(q)$ is the monic polynomial of the smallest degree over $\gf(q)$ with $\beta^s$ 
as a root.  It is straightforward to see that this polynomial is given by 
\begin{eqnarray}
\m_{\beta^s}(x)=\prod_{i \in C_s} (x-\beta^i) \in \gf(q)[x], 
\end{eqnarray} 
which is irreducible over $\gf(q)$. It then follows from (\ref{eqn-cosetPP}) that 
\begin{eqnarray}\label{eqn-canonicalfact}
x^n-1=\prod_{s \in  \Gamma_{(n,q)}} \m_{\beta^s}(x)
\end{eqnarray}
which is the factorization of $x^n-1$ into irreducible factors over $\gf(q)$. This canonical factorization of $x^n-1$ 
over $\gf(q)$ is crucial for the study of cyclic codes.

Let $\delta$ be an integer with $2 \leq \delta \leq n$ and let $h$ be an integer.  
A \emph{BCH code\index{BCH codes}} over $\gf(q)$ 
with length $n$ and \emph{designed distance} $\delta$, denoted by $\C_{(q,n,\delta,h)}$, is a cyclic code with 
generator polynomial 
\begin{eqnarray}\label{eqn-BCHdefiningSet}
g_{(q,n,\delta,h)}=\lcm(\m_{\beta^h}(x), \m_{\beta^{h+1}}(x), \cdots, \m_{\beta^{h+\delta-2}}(x)) 
\end{eqnarray}
where the least common multiple is computed over $\gf(q)$.

It may happen that $\C_{(q,n,\delta_1,h)}$ and $\C_{(q,n,\delta_2,h)}$ are identical for two distinct 
$\delta_1$ and $\delta_2$. The 
maximum designed distance of a BCH code is also called the \emph{Bose distance\index{Bose distance}}. 

When $h=1$, the code $\C_{(q,n,\delta,h)}$ with the generator polynomial in (\ref{eqn-BCHdefiningSet}) is called a \emph{narrow-sense\index{narrow sense}} BCH code. If $n=q^m-1$, then $\C_{(q,n,\delta,h)}$ is referred to as a \emph{primitive\index{primitive BCH}} BCH code. 

BCH codes are a subclass of cyclic codes with interesting properties. In many cases BCH codes are the best linear codes. 
For example, among all binary cyclic codes of odd length $n$ with $n \leq 125$ the best cyclic code is always a BCH code 
except for two special cases \cite{Dingbook15}. Reed-Solomon codes are also BCH codes and are widely used in communication 
devices and consumer electronics. In the past ten years, a lot of progress on the study of BCH codes has been made 
(see, for example, \cite{LWL19,LiSIAM,LLFLR,SYW,YLLY}).  

It is well known that the extended code $\overline{\C_{(q,q^m-1,\delta,1)}}$ of the narrow-sense primitive BCH code 
$\C_{(q,q^m-1,\delta,1)}$ holds $2$-designs, as the permutation automorphism group of the extended code contains 
the general affine group as a subgroup (see, for example, 
\cite{DingZhouConf17} and \cite[Chapter 8]{Dingbook18}). However, It is extremely rare that an infinite family of 
cyclic codes hold an infinite family of $3$-designs. In this paper, we will present an infinite family of BCH codes 
holding an infinite family of $3$-designs.

\section{Almost MDS codes and near MDS codes} 

The \emph{Singleton defect}\index{Singleton defect} of an $[n, k, d]$ code $\C$ is defined 
by $\sdef(\C)=n-k+1-d$. Thus, MDS codes are codes with defect $0$. A code $\C$ is said to be 
\emph{almost MDS} (AMDS for short) if it has defect $1$. Hence, AMDS codes have parameters 
$[n, k, n-k]$. AMDS codes of dimension $1$, $n-2$, $n-1$ and $n$ are called \emph{trivial}. 
Since it is easy to construct trivial AMDS codes of arbitrary lengths, we will consider only 
nontrivial AMDS codes. 

The following theorem summarises some basic properties of AMDS codes 
(see  \cite{DodLan95} and \cite{FaldumWillems97} for a proof). 

\begin{theorem}\label{thm-aug291}
Let $\C$ be an $[n, k, n-k]$ AMDS code over $\gf(q)$. 
\begin{itemize} 
\item If $k \geq 2$, then $n \leq k + 2q$. 
\item If $k \geq 2$ and $n-k > q$, then $k \leq 2q$. 
\item If $n-k>q$, then $\C^\perp$ is also AMDS.  
\item If $k \geq 2$, then $\C$ is generated by its codewords of weight $n-k$ and $n-k+1$. 
\item If $k \geq 2$ and $n-k >q$, then $\C$ is generated by its minimum weight codewords. 
\end{itemize} 
\end{theorem} 

Unlike MDS codes, the dual of an AMDS code may not be AMDS. 
A code $\C$ is said to be \emph{near MDS}\index{near MDS code} (NMDS\index{NMDS code} for short) 
if both $\C$ and $\C^\perp$ are AMDS. By definition, $\C$ is near NMDS if and only if $\C^\perp$ is so. 
The next theorem follows from the definition of NMDS codes. 

\begin{theorem}
An $[n, k]$ code $\C$ over $\gf(q)$ is NMDS if and only if $d(\C) + d(\C^\perp)=n$, where $d(\C)$ and 
$d(\C^\perp)$ denote the minimum distance of $\C$ and $\C^\perp$, respectively. 
\end{theorem} 

The next theorem gives a characterisation of near MDS codes \cite{DodLan95}. 

\begin{theorem}\label{thm-19aug282} 
An $[n,k]$ code $\C$ over $\gf(q)$ is NMDS if and only if a parity-check matrix, say $H$, 
(and consequently every parity-check matrix of $\C$) satisfies the following conditions: 
\begin{enumerate} 
\item any $n-k-1$ colmuns of $H$ are linearly independent; 
\item there exist $n-k$ linearly dependent columns; and 
\item any $n-k+1$ columns of $H$ are of rank $n-k$.  
\end{enumerate}   
\end{theorem} 

Theorem \ref{thm-19aug282} is equivalent to the following. 

\begin{theorem}\label{thm-19aug283} 
An $[n,k]$ code $\C$ over $\gf(q)$ is NMDS if and only if a generator matrix, say $G$, 
(and consequently every generator matrix of $\C$) satisfies the following conditions: 
\begin{enumerate} 
\item any $k-1$ colmuns of $G$ are linearly independent; 
\item there exist $k$ linearly dependent columns; and 
\item any $k+1$ columns of $G$ are of rank $k$.  
\end{enumerate}   
\end{theorem} 

We have the following weight distribution formulas for NMDS codes. 

\begin{theorem}[\cite{DodLan95}]\label{thm-DLwtd}
Let $\C$ be an $[n, k, n-k]$ NMDS code. Then the weight distributions of $\C$ and $\C^\perp$ 
are given by 
\begin{eqnarray}\label{eqn-DL281}
A_{n-k+s} = \binom{n}{k-s} \sum_{j=0}^{s-1} (-1)^j \binom{n-k+s}{j}(q^{s-j}-1) + 
             (-1)^s \binom{k}{s}A_{n-k}
\end{eqnarray} 
for $s \in \{1,2, \ldots, k\}$, and 
\begin{eqnarray}\label{eqn-DL282}
A_{k+s}^\perp = \binom{n}{k+s} \sum_{j=0}^{s-1} (-1)^j \binom{k+s}{j}(q^{s-j}-1) + 
             (-1)^s \binom{n-k}{s}A_{k}^\perp 
\end{eqnarray} 
for $s \in \{1,2, \ldots, n-k\}$. 
\end{theorem} 

Note that $\sum_{i=0}^n A_i=q^{k}$ and $\sum_{i=0}^n A_i^\perp=q^{n-k}$. The $A_{n-k}$ 
in \eqref{eqn-DL281} and $A_k^\perp$ in \eqref{eqn-DL282} cannot be determined by these 
two equations and Equations (\ref{eqn-DL281}) and (\ref{eqn-DL282}). It is possible that 
two $[n, k, n-k]$ NMDS codes over $\gf(q)$ have different weight distributions.  Thus, 
the weight distribution of an $[n, k, n-k]$ NMDS code over $\gf(q)$ depends on not only  $n$, 
$k$ and $q$, but also some other parameters of the code. This is a major difference between MDS 
codes and NMDS codes. 

\begin{example} 
The $[11,6,5]$ ternary Golay code $\C_3$ has weight enumerator 
$$ 
1 +  132z^5 + 132z^6 + 330z^8 + 110z^9 + 24z^{11}.  
$$  
Take a generator matrix $G_3$ of the ternary Golay code $\C_3$. When $G_3$ is viewed 
as a matrix over $\gf(9)$, it generates a linear code $\C_9$ over $\gf(9)$ with parameters 
$[11,6,5]$ and weight enumerator 
$$ 
1+ 528z^5 + 528z^6 + 15840z^7 + 40920z^8 + 129800z^9 +  198000z^{10} + 145824z^{11}.  
$$ 
The dual code $\C_9^\perp$ has parameters $[11,5,6]$. Hence, the code $\C_9$ over $\gf(9)$ 
is NMDS. 

The extended code $\overline{\C_{(9,10,3,1)}}$ of the narrow-sense BCH code $\C_{(9,10,3,1)}$ 
over $\gf(9)$ has parameters $[11,6,5]$ and weight enumerator 
$$ 
1+ 240z^5 + 2256z^6 + 11520z^7 + 46680z^8 + 125480z^9 + 199728z^{10} + 145536z^{11}. 
$$
Its dual has parameters $[11,5,6]$. Thus,  $\overline{\C_{(9,10,3,1)}}$ is an NMDS code over 
$\gf(9)$, which has the same parameters as $\C_9$. However, the two codes have distinct 
weight enumerators. 
\end{example}

It follows from Theorem \ref{thm-DLwtd} that $A_{n-k}=A_k^\perp$ for any $[n, k, n-k]$ NMDS code. Consequently,  
any $[2k,k, k]$ NMDS code $\C$ and its dual are formally self-dual. 

The following is a corollary of Theorem \ref{thm-DLwtd}  \cite{DodLan95}. 

\begin{corollary}\label{cor-DL95} 
For an $[n,k, n-k]$ NMDS code over $\gf(q)$, it holds 
\begin{eqnarray}\label{eqn-DL283}
A_{n-k} \leq \binom{n}{k-1}\frac{q-1}{k}, 
\end{eqnarray} 
with equality if and only if $A_{n-k+1}=0$. By duality, 
\begin{eqnarray}
A_{k}^\perp \leq \binom{n}{k+1}\frac{q-1}{n-k}, 
\end{eqnarray}\label{eqn-DL284}  
with equality if and only if $A_{k+1}^\perp=0$. 
\end{corollary} 

For an $[n,k, n-k]$ NMDS code over $\gf(q)$, we deduce from (\ref{eqn-DL283}) that 
$$ 
\frac{A_{n-k}}{q-1} \leq \binom{n}{k-1}\frac{1}{k} = \binom{n}{n-k} \frac{n-k}{n-k+1} <  \binom{n}{n-k}.  
$$
Therefore, if the minimum weight codewords of an $[n, k, n-k]$ NMDS code over $\gf(q)$  support a 
$t$-design, then the $t$-design cannot be the complete design. 
  
It will be shown in Section \ref{sec-desNMDScodes}  that an $[n,k, n-k]$ NMDS code $\C$ over $\gf(q)$ with $A_{n-k+1}=0$ or 
$A_{k+1}^\perp=0$ yields $t$-designs for some positive integer $t$, and are thus very attractive. 
One basic question is whether such code exists. We will look into 
this existence problem in Section \ref{sec-desNMDScodes}. 

As a consequence of Corollary \ref{cor-DL95}, we have the following result \cite{DodLan95}. 

\begin{corollary} 
For any $[n,k,n-k]$ NMDS code over $\gf(q)$ with $A_{n-k+1}=0$, we have $k \leq n/2$. 
\end{corollary}     

NMDS codes have nice properties. In particular, up to a multiple, there is a natural 
correspondence between the minimum weight codewords of an NMDS code $\C$ and its dual 
$\C^\perp$, which follows from the next result \cite{FaldumWillems97}. 

\begin{theorem}\label{thm-121FW}
Let $\C$ be an NMDS code. Then for every minimum weight codeword $\bc$ in $\C$, there exists, 
up to a multiple, a unique minimum weight codeword $\bc^\perp$ in $\C^\perp$ such that 
$\support(\bc) \cap \support(\bc^\perp)=\emptyset$. In particular, $\C$ and $\C^\perp$ 
have the same number of minimum weight codewords. 
\end{theorem} 

By Theorem \ref{thm-121FW}, if the minimum weight codewords of an NMDS code support a $t$-design, 
so do the minimum weight codewords of its dual, and the two $t$-designs are complementary of each other. 
The following result says that infinite families of NMDS codes do exist. 

\begin{theorem}[\cite{TVlad}] 
Algebraic geometric $[n, k, n-k]$ NMDS codes over $\gf(q)$, $q=p^m$, do exist for every $n$ with 
\begin{eqnarray*}
n \leq 
\left\{ 
\begin{array}{ll}
q + \lceil 2 \sqrt{q} \rceil & \mbox{ if $p$ divides $\lceil 2 \sqrt{q} \rceil$ and $m$ is odd,} \\
q + \lceil 2 \sqrt{q} \rceil +1 & \mbox{ otherwise,}   
\end{array}
\right. 
\end{eqnarray*}
and arbitrary $k \in \{2,3, \ldots, n-2\}$. 
\end{theorem} 

Further information on near MDS codes from algebraic curves can be found in \cite{AL05}. 
While it is easy to construct NMDS codes, we are interested in only NMDS codes holding $t$-designs. 
We will treat such NMDS codes in Section \ref{sec-desNMDScodes}.

We now introduce extremal NMDS codes. Let $n(k, q)$ denote the maximum possible length of 
an NMDS code of fixed dimension $k$ over a fixed field $\gf(q)$. Then we have the following \cite{DodLan95}. 

\begin{lemma}\label{lem-extremalNMDScode} 
Let notation be the same as before. Then $n(k, q) \leq 2q+k$. For any $[2q+k, k, 2q]$ NMDS code over $\gf(q)$, $A_{2q+1}=0$. 
\end{lemma}

An NMDS code meeting the bound of Lemma \ref{lem-extremalNMDScode} is said to be \emph{extremal}\index{extremal NMDS code}, i.e., any $[2q+k, k, 2q]$ NMDS code over $\gf(q)$ is \emph{extremal}. 
The dual and the extended code of the $[11, 6, 5]$ ternary Golay code are extremal.  
NMDS codes with parameters $[2q+k, k+1, 2q-1]$ are said to be \emph{almost extremal}. 

\begin{theorem}[\cite{DeBoer96}]\label{thm-DeB291}
If $\C$ is a $[2q+k, k, 2q]$ extremal NMDS code over $\gf(q)$ with $k>q$, then $\C$ must be the dual of one 
of the following codes: 
\begin{itemize} 
\item the $[7,3,4]$ Hamming code over $\gf(2)$; 
\item the $[8,4,4]$ extended Hamming code over $\gf(2)$ (which is self-dual);
\item a $[10,6,4]$ punctured Golay code over $\gf(3)$; 
\item the $[11,6,5]$ Golay code over $\gf(3)$; and 
\item the $[12,6,6]$ extended Golay code over $\gf(3)$. 
\end{itemize} 
\end{theorem}  

In spite of Theorem \ref{thm-DeB291}, $[2q+k, k, 2q]$ extremal NMDS codes over $\gf(q)$ 
with $k \leq q$ may exist.  It will be shown that extremal NMDS codes yield $t$-designs for 
some $t$. Thus, we are very much fond of extremal NMDS codes.

\begin{theorem}[\cite{DeBoer96}]\label{thm-DeB292}
If $\C$ is a $[2q+k, k+1, 2q-1]$ almost extremal NMDS code over $\gf(q)$ with $k \geq q$, then $\C$ must be the dual of one 
of the following codes: 
\begin{itemize} 
\item a $[6,3,3]$ punctured Hamming code over $\gf(2)$; 
\item the $[7,3,4]$ Simplex code over $\gf(2)$;
\item a $[9,5,4]$ shortened punctured Golay code over $\gf(3)$; and  
\item a $[10,5,5]$ shortened Golay code over $\gf(3)$.  
\end{itemize} 
\end{theorem}  

In spite of Theorem \ref{thm-DeB292}, $[2q+k, k+1, 2q-1]$ almost extremal NMDS codes over $\gf(q)$ 
with $k < q$ may exist. It is open if almost extremal NMDS codes hold $t$-designs in general or not.  
Information about near MDS codes over small fields can be found in \cite{DodLan00}.

\section{Infinite families of near MDS codes holding $t$-designs}\label{sec-desNMDScodes} 

In this section, for the first time we present  infinite families of NMDS codes 
holding an infinite family of $t$-designs for $t \geq 2$, and settle a 70-year-old problem.

\subsection{A general theorem about $t$-designs from NMDS codes} 

First of all, we point out that some NMDS codes do not hold simple designs at all. Below is an example. 

\begin{example} 
The extended code $\overline{\C_{(9,10,3,1)}}$ of the narrow-sense BCH code $\C_{(9,10,3,1)}$ 
over $\gf(9)$ has parameters $[11,6,5]$. Its dual has parameters $[11,5,6]$. Both 
$\overline{\C_{(9,10,3,1)}}$ and its dual $\overline{\C_{(9,10,3,1)}}^\perp$ do not hold 
simple $1$-designs according to our Magma computations. 
\end{example} 

Secondly, some NMDS codes may hold $t$-designs. The following theorem is interesting, 
as it tells us that near MDS codes could hold $t$-designs \cite{DodLan95}.  

\begin{theorem}\label{thm-DLdesign}
Let $\C$ be an $[n, k, n-k]$ NMDS code over $\gf(q)$. If there exists an integer $s \geq 1$ 
such that $A_{n-k+s}=0$. Then the supports of the codewords of weight $k$ in $\C^\perp$ form 
a $(k-s)$-design. In particular, the supports of the codewords of minimal weight in the dual 
of an extremal NMDS code form a Steiner system $S(k-1, k, 2q+k)$. 
\end{theorem}

To apply Theorem \ref{thm-DLdesign}, one has to find NMDS codes satisfying $A_{n-k+s}=0$ for 
some $s<k$. The $[11,6,5]$ ternary Golay code and its duals as well as their extended codes 
are such NMDS codes. There are also several examples of such NMDS codes. But we are really 
interested in infinite families of such NMDS codes. 
In Section \ref{sec-designNMDSCode}, we will present such infinite family of ternary codes.  

\subsection{Infinite families of near MDS codes holding infinite families of $t$-designs}\label{sec-designNMDSCode} 

Throughout this section, let $p=p^s$, where $p$ is a prime and $s$ is a positive integer. In this section, we consider the 
narrow-sense BCH code $\C_{(q, q+1, 3, 1)}$ over $\gf(q)$ and its dual, and prove that they are near MDS and 
hold $3$-designs when $p=3$ and $2$-designs when $p=2$ and $s$ is even. 

We will need the following lemma whose proof is straightforward. 

\begin{lemma}\label{lem-tcm101}
Let $x, y, z \in \gf(q^2)^*$. Then 
\begin{eqnarray*}
\left|
\begin{array}{lll} 
x^{-1} & y^{-1} & z^{-1} \\ 
x  & y & z \\
x^2 & y^2 & z^2 
\end{array} 
\right| 
= 
\frac{(x-y)(y-z)(z-x)}{xyz} (xy+yz+zx).  
\end{eqnarray*} 
\end{lemma} 

We will also need the following lemma shortly. 

\begin{lemma}\label{lem-tcm102}  
Let $U_{q+1}$ denote the set of all $(q+1)$-th roots of unity in $\gf(q^2)$. Suppose that 
$x, y, z$ are three pirwise distinct elements in $U_{q+1}$ such that 
\begin{eqnarray}\label{eqn-tcm0}
\left|
\begin{array}{lll} 
x^{-1} & y^{-1} & z^{-1} \\ 
x  & y & z \\
x^2 & y^2 & z^2 
\end{array} 
\right| 
=0. 
\end{eqnarray}
Then $(x/y)^3=1$, which implies that $3$ divides $q+1$. 
\end{lemma} 

\begin{proof}
It follows from Lemma \ref{lem-tcm101} that 
\begin{eqnarray}\label{eqn-tcm101}
xy+yz+zx=0. 
\end{eqnarray} 
Raising both sides of (\ref{eqn-tcm0}) to the $q$-th power yields 
\begin{eqnarray}\label{eqn-tcm102}
\left|
\begin{array}{lll} 
x^{-q} & y^{-q} & z^{-q} \\ 
x^q  & y^q & z^q \\
x^{2q} & y^{2q} & z^{2q}  
\end{array} 
\right| 
=0. 
\end{eqnarray} 
Notice that $x, y, z \in U_{q+1}$. Equation (\ref{eqn-tcm102}) is the same as 
\begin{eqnarray}\label{eqn-tcm103}
\left|
\begin{array}{lll} 
x & y & z \\ 
x^{-1}  & y^{-1} & z^{-1} \\
x^{-2} & y^{-2} & z^{-2}  
\end{array} 
\right| 
=0. 
\end{eqnarray} 
It then follows from Lemma \ref{lem-tcm101} that (\ref{eqn-tcm103}) that 
$$ 
0=\frac{1}{xy}+\frac{1}{yz}+\frac{1}{zx}=\frac{x+y+z}{xyz}. 
$$
Consequently, 
\begin{eqnarray}\label{eqn-tcm104} 
x+y+z=0. 
\end{eqnarray} 
Combining (\ref{eqn-tcm101}) and (\ref{eqn-tcm104}) gives that $x^2+xy+y^2=0$. Thus $x^3=y^3$ 
and $(x/y)^3=1$. Note that $(x/y)^{q+1}=1$ and $x/y \neq 1$. We deduce that $3$ divides $q+1$. 
This completes the proof. 
\end{proof}

We are now ready to prove the following result about the code $\C_{(q, q+1, 3,1)}$. 

\begin{theorem}\label{thm-SQScode0}  
Let $q=p^s \geq 5$ with $s$ being a positive integer. Then the narrow-sense BCH code $\C_{(q, q+1, 3,1)}$ over $\gf(q)$ 
has parameters $[q+1, q-3, d]$, where $d=3$ if $3$ divides $q+1$ and $d \geq 4$ if $3$ does not divide $q+1$. 
\end{theorem} 

\begin{proof}
Put $n=q+1$. 
Let $\alpha$ be a generator of $\gf(q^2)^*$ and $\beta=\alpha^{q-1}$. Then $\beta$ is an $n$-th root of unity 
in $\gf(q^2)$. Let $g_1(x)$ and $g_2(x)$ denote the minimal polynomial of $\beta$ and $\beta^2$ over $\gf(q)$, 
respectively. 
Note that $g_1(x)$ has only roots $\beta$ and $\beta^q$ and  $g_2(x)$ has roots $\beta^2$ and $\beta^{q-1}$. 
One deduces that $g_1(x)$ and $g_2(x)$ are distinct irreducible polynomials of degree $2$. 
By definition, $g(x):=g_1(x)g_2(x)$ is the generator polynomial of $\C_{(q, q+1, 3,1)}$. Therefore, 
the dimension of $\C_{(q, q+1, 3,1)}$ is $q+1-4$. 
Note that 
$g(x)$ has only roots $\beta, \beta^2, \beta^{q-1}$ and $\beta^q$. By the BCH bound, the minimum weight 
of $\C_{(q, q+1, 3,1)}$ is at least $3$. 
Put $\gamma=\beta^{-1}$.  Then $\gamma^{q+1}=\beta^{-(q+1)}=1$. 
It then follows from Delsarte's theorem that the trace expression of $\C_{(q, q+1, 3,1)}^\perp$ is given by 
\begin{eqnarray}
\C_{(q, q+1, 3,1)}^\perp=\{\bc_{(a,b)}: a, b \in \gf(q^2)\}, 
\end{eqnarray} 
where $\bc_{(a,b)}=(\tr_{q^2/q}(a\gamma^i +b \gamma^{2i}))_{i=0}^q$. 

Define 
\begin{eqnarray}
H=\left[ 
\begin{array}{rrrrr}
1  & \gamma^1 & \gamma^2 & \cdots & \gamma^q \\
1  & \gamma^2 & \gamma^4 & \cdots & \gamma^{2q} 
\end{array}
\right].  
\end{eqnarray} 
It is easily seen that $H$ is a parity-check matrix of $\C_{(q, q+1, 3,1)}$, i.e., 
$$ 
\C_{(q, q+1, 3,1)}=\{\bc \in \gf(q)^{q+1}: \bc H^T=\bzero\}. 
$$

Assume that $3$ divides $q+1$. Notice that $\gamma$ is a primitive $(q+1)$-th root of unity. It then follows 
from $\gamma^{q+1}=1$ that 
\begin{eqnarray}\label{eqn-19oct171}
1+ \gamma^{(q+1)/3} +  \gamma^{2(q+1)/3} = 0.
\end{eqnarray} 
Define a vector $\bc=(c_0, c_1, \ldots, c_q) \in \gf(q)^{q+1}$, where 
\begin{eqnarray*}
c_i=\left\{ 
\begin{array}{ll}
1 & \mbox{ if } i \in \left\{0, \frac{q+1}{3},  \frac{2(q+1)}{3} \right\}, \\
0 & \mbox{ otherwise.} 
\end{array}
\right. 
\end{eqnarray*}
Then $\bc \in \C_{(q, q+1, 3,1)}$. Consequently, $d=3$. 

Assume now that $3$ does not divide $q+1$. We only need to prove that $d \neq 3$. On the contrary, suppose 
$d=3$. Then there are three pairwise distinct elements $x, y, z$ in $U_{q+1}$ such that 
\begin{eqnarray}\label{eqn-tchm105}
a \left[ \begin{array}{c}
x \\
x^2
\end{array} 
\right] 
+ 
b \left[ \begin{array}{c}
y \\
y^2
\end{array} 
\right] 
+ 
c \left[ \begin{array}{c}
z \\
z^2
\end{array} 
\right] 
=0,  
\end{eqnarray} 
where $a, b, c \in \gf(q)^*$. 
Raising to the $q$-th power both sides of the equation $ax+by+cz=0$ yields 
\begin{eqnarray}\label{eqn-tchm106}
ax^{-1}+by^{-1}+cz^{-1}=0. 
\end{eqnarray}
Combining (\ref{eqn-tchm105}) and (\ref{eqn-tchm106}) gives 
\begin{eqnarray*}
a \left[ \begin{array}{c}
x^{-1} \\ 
x \\
x^2
\end{array} 
\right] 
+ 
b \left[ \begin{array}{c} 
y^{-1} \\ 
y \\
y^2
\end{array} 
\right] 
+ 
c \left[ \begin{array}{c} 
z^{-1} \\ 
z \\
z^2
\end{array} 
\right] 
=0. 
\end{eqnarray*} 
It then follows that 
\begin{eqnarray*}
\left|
\begin{array}{lll} 
x^{-1} & y^{-1} & z^{-1} \\ 
x  & y & z \\
x^2 & y^2 & z^2 
\end{array} 
\right| 
=0. 
\end{eqnarray*}
By Lemma \ref{lem-tcm102}, $3$ divides $q+1$. This is contrary to our assumption that $3$ does not divide $q+1$.  
This completes the proof. 
\end{proof}

The theorem below makes a breakthrough in 70 years in the sense that it presents the first family of linear codes 
meeting the condition of Theorem \ref{thm-DLdesign} and holding an infinite family of $3$-designs.  

\begin{theorem}\label{thm-SQScode} 
Let $q=3^s$ with $s \geq 2$. Then the narrow-sense BCH code $\C_{(q, q+1, 3,1)}$ over $\gf(q)$ 
has parameters $[q+1, q-3, 4]$, and its dual code $\C_{(q, q+1, 3,1)}^\perp$ has parameters $[q+1, 4, q-3]$ 
and weight enumerator 
\begin{eqnarray*}
1+ \frac{(q-1)^2q(q+1)}{24} z^{q-3} + \frac{(q-1)q(q+1)(q+3)}{4} z^{q-1} + \\ 
\frac{(q^2-1)(q^2-q+3)}{3} z^q +  \frac{3(q-1)^2q(q+1)}{8} z^{q+1}. 
\end{eqnarray*} 
Further, the minimum weight codewords in $\C_{(q, q+1, 3,1)}^\perp$ support a $3$-$(q+1, q-3, \lambda)$ design 
with 
$$ 
\lambda=\frac{(q-3)(q-4)(q-5)}{24}, 
$$ 
and the minimum weight codewords in $\C_{(q, q+1, 3,1)}$ support a $3$-$(q+1, 4, 1)$ design, i.e., 
a Steiner quadruple system\index{Steiner quadruple system} $S(3,4,3^s+1)$. Furthermore, 
the codewords of weight 5 in $\C_{(q, q+1, 3,1)}$ support a $3$-$(q+1, 5, (q-3)(q-7)/2)$ design.  
\end{theorem}

\begin{proof}
We follow the notation of the proof of Theorem \ref{thm-SQScode0}. Since $3$ does not divide $q+1=3^s+1$, 
by Theorem \ref{thm-SQScode0} the minimum distance $d$ of $\C_{(q, q+1, 3,1)}$ is at least 4. We now 
prove that $d=4$ and the codewords of weight $4$ in $\C_{(q, q+1, 3,1)}$ support a $3$-$(q+1, 4, 1)$ design.  

Let $x, y, z$ be three pairwise distinct elements in $U_{q+1}$. We conclude that $x+y+z \neq 0$. Suppose on 
the contrary that $x+y+z=0$. We have then 
$$ 
0=(x+y+z)^q=x^q+y^q+z^q=\frac{1}{x}+\frac{1}{y}+\frac{1}{z}=\frac{xy+yz+zx}{xyz}. 
$$ 
In summary, we have  
\begin{eqnarray*}
\left\{ 
\begin{array}{l}
x+y+z=0, \\
xy+yz+zx=0,  
\end{array}
\right. 
\end{eqnarray*} 
which is the same as 
\begin{eqnarray*}
\left[ 
\begin{array}{l}
x \\
x^2 
\end{array}
\right] + 
\left[ 
\begin{array}{l}
y \\
y^2 
\end{array}
\right] +
\left[ 
\begin{array}{l}
z \\
z^2 
\end{array}
\right] = \bzero.   
\end{eqnarray*}
This means that $\C_{(q, q+1, 3,1)}$ has a codeword of weight $3$, which is contrary to Theorem  \ref{thm-SQScode0}. 

We now prove that there is a unique $w \in U_{q+1} \setminus \{x, y, z\}$ such that 
\begin{eqnarray}\label{eqn-tcm161}
\lefteqn{\left|
\begin{array}{llll} 
x^{-2} & y^{-2} & z^{-2}  & w^{-2} \\ 
x^{-1}  & y^{-1} & z^{-1} & w^{-1} \\ 
x & y & z & w \\  
x^{2} & y^{2} & z^{2}  & w^{2} 
\end{array} 
\right|  }  \nonumber \\
&& = \frac{(z-w)(y-w)(y-z)(x-w)(x-z)(x-y)}{(xyzw)^2} (xy+xz+xw+yz+yw+zw)  \nonumber \\ 
&& =0. 
\end{eqnarray} 
Note that $(z-w)(y-w)(y-z)(x-w)(x-z)(x-y) \neq 0$. It follows from (\ref{eqn-tcm161}) that 
\begin{eqnarray}\label{eqn-tcm162}
w=-\frac{xy+yz+zx}{x+y+z}. 
\end{eqnarray} 
We need to prove that $w \in U_{q+1}$. Note that 
$$ 
w=- xyz \frac{\frac{1}{x} + \frac{1}{y} + \frac{1}{z}}{x+y+z}= - xyz \frac{(x+y+z)^q}{x+y+z}. 
$$ 
We then have 
$$ 
w^{q+1}=(- xyz)^{q+1} (x+y+z)^{q^2-1} =1. 
$$ 
By definition,  $w \in U_{q+1}$. 

We now prove that $w \neq x$. Suppose on the contrary that $w=x$, then 
$$ 
x = -\frac{xy+yz+zx}{x+y+z}, 
$$ 
which yields 
$$ 
(x-z)(x-y)=0. 
$$ 
Whence, $x=z$ or $x=z$, which is contrary to our assumption that $x, y, z$ are three pairwise 
distinct elements in $U_{q+1}$. Due to symmetry, $w \neq y$ and $w \neq z$. 
The uniqueness of $w$ is justified by (\ref{eqn-tcm162}). 

Note that $U_{q+1}=\{1, \gamma, \gamma^2, \cdots, \gamma^q\}$. Let $\{x, y, z, w\}$ be 
any 4-subset of $U_{q+1}$ such that (\ref{eqn-tcm161}) holds. Without loss of generality, 
assume that 
$$ 
x=\gamma^{i_1}, \ y=\gamma^{i_2}, \ z=\gamma^{i_3},  \ w=\gamma^{i_4},    
$$ 
where $0 \leq i_1<i_2<i_3<i_4 \leq q$.  Since $d \geq 4$, the rank of the matrix 
\begin{eqnarray}\label{eqn-tcm164} 
M(x, y, z, w):=\left[ 
\begin{array}{llll} 
x^{-2} & y^{-2} & z^{-2}  & w^{-2} \\ 
x^{-1}  & y^{-1} & z^{-1} & w^{-1} \\ 
x & y & z & w \\  
x^{2} & y^{2} & z^{2}  & w^{2} 
\end{array} 
\right] 
\end{eqnarray} 
equals $3$. Let $(u_{i_1}, u_{i_2}, u_{i_3}, u_{i_4})$ denote a nonzero solution of 
\begin{eqnarray*} 
\left[ 
\begin{array}{llll} 
x^{-2} & y^{-2} & z^{-2}  & w^{-2} \\ 
x^{-1}  & y^{-1} & z^{-1} & w^{-1} \\ 
x & y & z & w \\  
x^{2} & y^{2} & z^{2}  & w^{2} 
\end{array} 
\right] 
\left[ 
\begin{array}{llll} 
u_{i_1} \\ 
u_{i_2} \\ 
u_{i_3} \\  
u_{i_4}  
\end{array} 
\right] 
= \bzero. 
\end{eqnarray*} 
Since the rank of the matrix of (\ref{eqn-tcm164}) is $3$, all these $u_{i_j} \neq 0$.   
Define a vector $\bc=(c_0, c_1, \ldots, c_q) \in \gf(q)^{q+1}$, where $c_{i_j}=u_{i_j}$ for $j \in \{1,2,3,4\}$ 
and $c_h =0$ for all $h \in \{0,1, \ldots, q\} \setminus \{i_1, i_2, i_3, i_4\}$. It is easily seen that $\bc$ is a 
codeword of weight $4$ in $\C_{(q, q+1, 3,1)}$. The set $\{a\bc: a \in \gf(q)^*\}$ consists of all such codewords 
of weight $4$ with nonzero coordinates in $\{i_1, i_2, i_3, i_4\}$. Hence, $d=4$. Conversely, every codeword of weight $4$ in $\C_{(q, q+1, 3,1)}$ with 
nonzero coordinates in $\{i_1, i_2, i_3, i_4\}$ must correspond to the set $\{x,y,z,w\}$. Hence, every codeword 
of weight $4$ and its nonzero multiples in $\C_{(q, q+1, 3,1)}$ correspond to such set $\{x,y,z,w\}$ uniquely. 
We then deduce that the codewords of weight $4$ in  $\C_{(q, q+1, 3,1)}$ support a $3$-$(q+1, 4, 1)$ 
design. As a result,  
$$ 
A_4=(q-1) \frac{\binom{q+1}{3}}{\binom{4}{3}} = \frac{(q-1)^2q(q+1)}{24}.
$$

Note that $\C_{(q, q+1, 3,1)}$ has parameters $[q+1, q-3, 4]$. We now prove that the minimum distance $d^\perp$ 
of $\C_{(q, q+1, 3,1)}^\perp$ is equal to $q-3$. Recall that 
\begin{eqnarray*}
\C_{(q, q+1, 3,1)}^\perp=\{\bc_{(a,b)}: a, b \in \gf(q^2)\}, 
\end{eqnarray*} 
where $\bc_{(a,b)}=(\tr_{q^2/q}(a\gamma^i +b \gamma^{2i}))_{i=0}^q$. 
Let $u \in U_{q+1}$. Then 
$$ 
\tr_{q^2/q}(au+bu^2)=au+bu^2 + a^qu^{-1}+b^q u^{-2} 
=u^{-2}(bu^4+au^3 +a^qu+b^q). 
$$ 
Hence, there are at most four $u \in U_{q+1}$ such that $\tr_{q^2/q}(au+bu^2)=0$ if $(a, b) \neq (0,0)$. 
As a result, for $(a, b) \neq (0,0)$ 
we have 
$$ 
\wt(\bc_{(a,b)}) \geq q+1-4=q-3.  
$$ 
This means that $d^\perp \geq q-3$. If $d^\perp=q-2$, then $\C_{(q, q+1, 3,1)}^\perp$ would be an MDS code and  
$\C_{(q, q+1, 3,1)}$ would also be an MDS code, which leads to a contradiction. We then conclude that 
$d^\perp=q-3$. Now both $\C_{(q, q+1, 3,1)}$ and its dual are AMDS. By definition, both $\C_{(q, q+1, 3,1)}$ 
and its dual are NMDS. 
It then follows from Theorem \ref{thm-121FW} that 
$$ 
A_{q-3}^\perp =A_4= \frac{(q-1)^2q(q+1)}{24}. 
$$ 
Applying Theorem \ref{thm-DLwtd}, one obtains the desired weight enumerator of  $\C_{(q, q+1, 3,1)}^\perp$. 
In particular, $A_{q-2}^\perp =0$. It then follows from the Assmus-Mattson Theorem that the minimum weight 
codewords in $\C_{(q, q+1, 3,1)}^\perp$ support a $3$-$(q+1, q-3, \lambda)$ design 
with 
$$ 
\lambda=\frac{(q-3)(q-4)(q-5)}{24}.  
$$ 
Again by Theorem \ref{thm-DLwtd}, 
$$ 
A_5=\binom{q+1}{q-4}(q-1)+(q-3)A_4 = \frac{(q-7)(q-3)(q-1)^2q(q+1)}{5!}. 
$$ 
It then follows from the Assmus-Mattson Theorem again that the codewords of weight 5 in $\C_{(q, q+1, 3,1)}$ 
support a $3$-$(q+1, 5, (q-3)(q-7)/2)$ simple design. 
\end{proof}

The proof of Theorem \ref{thm-SQScode} also proved the following theorem. 

\begin{theorem}
Let $q=3^s$ with $s \geq 2$. Let $\alpha$ be a generator of $\gf(q^2)^*$, and put $\gamma=\alpha^{-(q-1)}$. 
Define $U_{q+1}=\{1, \gamma, \gamma^2, \ldots, \gamma^q\}$ and 
$$ 
\cB=\left\{\{x,y,z,w\} \in \binom{U_{q+1}}{4}: xy+xz+xw+yz+yw+zw=0\right\}.  
$$  
Then $(U_{q+1}, \cB)$ is a Steiner system 
$S(3,4,3^s+1)$, and is isomorphic to the Steiner system supported by the minimum weight codewords of 
the code $\C_{(q, q+1, 3,1)}$. 
\end{theorem}

It is known that a Steiner quadruple system $S(3,4, v)$ exists if and only if $v \equiv 2, 4 \pmod{6}$ \cite{Hanani60}. 
There are two different constructions of an infinite family of Steiner systems $S(3, q+1, q^s+1)$ for $q$ being a prime 
power and $s \geq 2$. 
The first produces the spherical designs due to \citet{Witt38}, which is based on the action of $\PGL_2( \gf(q^s))$ on the base block 
$\gf(q) \cup \{\infty\}$. The automorphism group of the spherical design contains the group 
$\PGaL_2(\gf(q^s))$. The second construction was proposed in \cite{KeyWagner86}, and is based on affine spaces. 
The Steiner systems  $S(3, q+1, q^s+1)$ from the two constructions are not isomorphic \cite{KeyWagner86}. 
  
When $q=p^s$ for $p > 3$ and $3$ does not divide $q+1$,  the code  $\C_{(q, q+1, 3,1)}$ of Theorem \ref{thm-SQScode} is still  NMDS, but 
it does not hold $2$-designs according to Magma experiments. The case $q=3^s$ is really special. 
When $s \in \{2,3\}$, the Steiner quadruple system $S(3,4,3^s+1)$ of  Theorem \ref{thm-SQScode} 
is isomorphic to the spherical design with the same parameters. We conjecture that they are isomorphic 
in general, but do not have a proof. The first contribution  of Theorem \ref{thm-SQScode} is a coding-theoretic 
construction of the spherical quadruple systems  $S(3,4,3^s+1)$. 
The second contribution is that it presents the first infinite family of NMDS codes 
holding an infinite family of $3$-designs since the first NMDS ternary code discovered 70 years ago 
by \cite{Golay49}.

It was shown that  the total number of nonisomorphic cyclic Steiner quadruple systems $S(3, 4, 28)$ is $1028387$ \cite{CFFHO}, which is a big number. This number indicates that it is a hard problem to classify Steiner quadruple systems.

A family of NMDS codes may not satisfy the condition of Theorem \ref{thm-DLdesign} 
(i.e., the conditions in the Assmus-Mattson theorem), 
but could still hold $2$-designs. The next theorem introduces a family of such NMDS 
codes and their designs.   

\begin{theorem}\label{thm-SQScode2} 
Let $q=2^s$ with $s \geq 4$ being even. Then the narrow-sense BCH code $\C_{(q, q+1, 3,1)}$ over $\gf(q)$ 
has parameters $[q+1, q-3, 4]$, and its dual code $\C_{(q, q+1, 3,1)}^\perp$ has parameters $[q+1, 4, q-3]$ 
and weight enumerator 
\begin{eqnarray*}
1+ \frac{(q-4)(q-1)q(q+1)}{24}z^{q-3} +  \frac{(q-1)q(q+1)}{2} z^{q-2} + \frac{(q+1)q^2(q-1)}{4} z^{q-1} \\
+ 
\frac{(q-1)(q+1)(2q^2+q+6)}{6} z^q +  \frac{3q^4 - 4q^3 - 3q^2 + 4q}{8} z^{q+1}. 
\end{eqnarray*} 
Further,  the codewords of weight $4$ in $\C_{(q, q+1, 3,1)}$ support a $2$-$(q+1, 4, (q-4)/2)$ design, and 
the codewords of weight $q-3$ in the dual code $\C_{(q, q+1, 3,1)}^\perp$ support a $2$-$(q+1, q-3, \lambda^\perp)$ design 
with 
$$ 
\lambda^\perp=\frac{(q-4)^2(q-3)}{24}. 
$$ 
\end{theorem} 

\begin{proof} 
Recall that $q=2^s$ with $s \geq 4$. 
We follow the notation of the proof of Theorem \ref{thm-SQScode}.  
Let $x, y, z, w$ be four pairwise distinct elements 
in $U_{q+1}$. It can be verified that 
\begin{eqnarray}\label{eqn-tcm171}
\lefteqn{\left|
\begin{array}{llll} 
x^{-2} & y^{-2} & z^{-2}  & w^{-2} \\ 
x^{-1}  & y^{-1} & z^{-1} & w^{-1} \\ 
x & y & z & w \\  
x^{2} & y^{2} & z^{2}  & w^{2} 
\end{array} 
\right|  }  \nonumber \\
&& = \frac{(z-w)(y-w)(y-z)(x-w)(x-z)(x-y)}{(xyzw)^2} (xy+xz+xw+yz+yw+zw). \nonumber  \\ 
\end{eqnarray} 
Notice that $3$ does not divide $2^s+1$, as $s$ is even. 
It can be similarly proved that $\C_{(q, q+1, 3,1)}$ over $\gf(q)$ 
has parameters $[q+1, q-3, 4]$, and its dual code $\C_{(q, q+1, 3,1)}^\perp$ has parameters $[q+1, 4, q-3]$. 
Thus, they are NMDS. 

Similar to the proof of Theorem \ref{thm-SQScode}, one can prove that every codeword of weight 4 in  $\C_{(q, q+1, 3,1)}$ and its nonzero multiples correspond 
uniquely to a set $\{x,y,z,w\}$ of four pairwise distinct elements $x, y, z, w$ in $U_{q+1}$ such that the matrix 
$M(x,y,z,w)$ in  (\ref{eqn-tcm164} ) has rank $3$. 

Let $x, y$ be two distinct elements in $U_{q+1}$. We now consider the total number of choices of $z$ and $w$ in $U_{q+1}$ 
such the matrix $M(x,y,z,w)$  in  (\ref{eqn-tcm164} ) has rank $3$. Using (\ref{eqn-tcm171}) one can verify that $M(x,y,z,w)$ has rank $3$ if and only if 
\begin{eqnarray*} 
z \not\in \left\{  x, y, x^2y^{-1}, y^2x^{-1}, (xy)^{2^{2s-1}}   \right\} 
\end{eqnarray*} 
and 
\begin{eqnarray}\label{eqn-tcm172}
w=\frac{xy+yz+zx}{x+y+z}. 
\end{eqnarray} 
Note that the elements in 
$$ 
 \left\{  x, y, x^2y^{-1}, y^2x^{-1}, (xy)^{2^{2s-1}}   \right\} 
$$ 
are pairwise distinct. It can be verified that if $(z, w)$ is a choice, so is $(w, z)$. Thus, the total number of choices 
of $w$ and $z$ such the matrix $M(x,y,z,w)$  in  (\ref{eqn-tcm164} ) has rank $3$ is equal to 
$$ 
\frac{q+1-5}{2}=\frac{q-4}{2}. 
$$
Since this number is independent of the elements $x$ and $y$, the codewords of weight $4$ in $\C_{(q, q+1, 3,1)}$ support 
a $2$-$(q+1, 4, (q-4)/2)$ design. Consequently, 
$$ 
A_4 = \frac{(q-4)(q-1)q(q+1)}{24}. 
$$ 
It then follows from Theorem \ref{thm-121FW} that 
$$ 
A_{q-3}^\perp =A_4= \frac{(q-4)(q-1)q(q+1)}{24}.  
$$ 
Applying Theorem \ref{thm-DLwtd}, one obtains the desired weight enumerator of $\C_{(q, q+1, 3,1)}^\perp$. 

By Theorem \ref{thm-121FW}, the minimum weight codewords in $\C_{(q, q+1, 3,1)}^\perp$ support a 
$2$-design which is the complementary design of the design supported by all the minimum weight codewords 
in $\C_{(q, q+1, 3,1)}$. This completes the proof. 
\end{proof} 

With Theorem \ref{thm-DLwtd} and the expression of $A_4$, one can verify that $A_i>0$ for all $i$ with 
$4 \leq i \leq q+1$. Notice that $A_i^\perp >0$ for all $i$ with $q-3 \leq i \leq q+1$. The conditions in the 
Assmus-Mattson Theorem and the condition of Theorem \ref{thm-DLdesign} are not satisfied. But the 
codes still hold simple $2$-designs.    

We remark the code $\C_{(q, q+1, 3, 1)}$ supports simple $t$-designs for $t \geq 2$ only when $p=3$ 
or $p=2$ and $s$ is even. This makes this class of codes very special.  

When $s$ is odd and $q=2^s$, the code  $\C_{(q, q+1, 3,1)}$ is not near MDS, but its dual is still AMDS. 
We are not interested in  the code $\C_{(q, q+1, 3,1)}$ in this case. The parameters of this code and its dual 
in this case are given 
below. 

\begin{theorem}
Let $q=2^s$ with $s \geq 3$ being odd. Then the narrow-sense BCH code $\C_{(q, q+1, 3,1)}$ over $\gf(q)$ 
has parameters $[q+1, q-3, 3]$, and its dual code $\C_{(q, q+1, 3,1)}^\perp$ has parameters $[q+1, 4, q-3]$ 
\end{theorem} 

\begin{proof}
We follow the notation before. The dimensions of the two codes in this case follow from the proof of Theorem 
\ref{thm-SQScode2}.  
By the BCH bound, the minimum distance $d(\C_{(q, q+1, 3,1)}) \geq 3$. We now 
prove that the code has a codeword of weight $3$. Since $s$ is odd, $3$ divides $q+1$. It then follows that 
(\ref{eqn-19oct171}) 
 that $\C_{(q, q+1, 3,1)}$ has a codeword of weight $3$. Consequently,  $d(\C_{(q, q+1, 3,1)}) = 3$. 
The minimum distance of $\C_{(q, q+1, 3,1)}^\perp$ is similarly proved. 
\end{proof}

The proof of Theorem \ref{thm-SQScode2} also proved the following theorem. 

\begin{theorem}
Let $q=2^s$ with $s \geq 4$ being even. Let $\alpha$ be a generator of $\gf(q^2)^*$, and put $\gamma=\alpha^{-(q-1)}$. 
Define $U_{q+1}=\{1, \gamma, \gamma^2, \ldots, \gamma^q\}$ and 
$$ 
\cB=\left\{\{x,y,z,w\} \in \binom{U_{q+1}}{4}: xy+xz+xw+yz+yw+zw=0\right\}.  
$$  
Then $(U_{q+1}, \cB)$ is a $2$-$(q+1, 4, (q-4)/2)$ design, 
and is isomorphic to the $2$-design supported by the minimum weight codewords of 
the code $\C_{(q, q+1, 3,1)}$. 
\end{theorem}

\subsection{Subfield subcodes of the two families of near MDS codes} 

Let $\C$ be an $[n, \kappa, d]$ code over $\gf(q)$, where $q=r^h$ for some prime power $r$ and some positive integer $h$. 
The \emph{subfield subcode} of $\C$ over $\gf(r)$, denoted by $\C|_{\gf(r)}$, is defined by 
$$ 
\C|_{\gf(r)} = \{\bc \in \C: \bc \in \gf(r)^n\} = \C \cap \gf(r)^n. 
$$ 
It is well known that 
$$ 
\kappa \geq \dim(\C|_{\gf(r)}) \geq n-r(n-\kappa). 
$$ 
In this section, we provide information on the subfield subcodes of the two families of near MDS codes documented in 
Theorems \ref{thm-SQScode}  and \ref{thm-SQScode2}. 

Let $q=p^s$, where $p$ is a prime. We now consider the narrow-sense BCH code $\C_{(q, q+1, 3,1)}$ and its subfield subcode 
$\C_{(q, q+1, 3,1)}|_{\gf(p)}$. We follow the notation in the proof of Theorem \ref{thm-SQScode0}. By the  
Deslsarte Theorem, we have 
\begin{eqnarray}\label{eqn-1910151}
\C_{(q, q+1, 3,1)}|_{\gf(p)} = \left(  \tr_{q/p} \left( \C_{(q, q+1, 3,1)}^\perp \right) \right)^\perp. 
\end{eqnarray}    
By the proof of  Theorem \ref{thm-SQScode0}, 
\begin{eqnarray}\label{eqn-1910152}
  \C_{(q, q+1, 3,1)}^\perp 
 = \{ (\tr_{q^2/q} (a \gamma^i +b\gamma^{2i}))_{i=0}^q: a, \, b \in \gf(q^2) \}. 
\end{eqnarray} 
Combining (\ref{eqn-1910151}) and (\ref{eqn-1910152}) yields 
\begin{eqnarray*}
\C_{(q, q+1, 3,1)}|_{\gf(p)} = \left(\{ (\tr_{q^2/p} (a \gamma^i +b\gamma^{2i}))_{i=0}^q: a, \, b \in \gf(q^2) \} \right)^\perp. 
\end{eqnarray*} 
Again by the Delsarte Theorem, we obtain 
\begin{eqnarray}\label{eqn-1910153}
\C_{(q, q+1, 3,1)}|_{\gf(p)} = \C_{(p, q+1, 3,1)}. 
\end{eqnarray} 
This equality will be useful for deriving the parameters of the subfield subcode. 

\begin{theorem}\label{thm-1910155}
Let  $s \geq 4$ be an even integer. Then the binary subfield subcode $\C_{(2^s, 2^s+1, 3, 1)}|_{\gf(2)}$  
has parameters $[2^s+1, 2^s+1-2s, 5]$.  
\end{theorem} 

\begin{proof}
Let $q=2^s$ and $n=q+1=2^s+1$. We follow the notation in the proofs of Theorems \ref{thm-SQScode0} and \ref{thm-SQScode2}.  
Note that the $2$-cyclotomic coset $C_1$ modulo $n$ is given by 
$$ 
C_1=\{1,2, \ldots, 2^{s-1}, -1, -2, \ldots, -2^{s-1}\} \bmod n. 
$$ 
By definition, the minimal polynomial $\M_{\beta}(x)$ of $\beta$ over $\gf(2)$ is given by 
$$ 
\M_{\beta}(x)=\sum_{i \in C_1} (x - \beta^i). 
$$
By the definition of BCH codes, $\C_{(2, 2^s+1, 3,1)}$ has generator polynomial $\M_{\beta}(x)$ with degree $2s$, 
and is the Zetterberg code. It is known that this code has minimum distance $5$ \cite{SV91,Xiaetal}. 
\end{proof} 

Using the sphere packing bound, one can verify that the subfield subcode $\C_{(2^s, 2^s+1, 3, 1)}|_{\gf(2)}$ 
is dimension-optimal. In addition,  this  binary code is also distance-optimal when $s \in \{2,4, 6, 8\}$ according to 
\cite{Grassl}. This makes the original code  $\C_{(2^s, 2^s+1, 3, 1)}$ very interesting. 

We inform the reader that $d(\C_{(2^s, 2^s+1, 3, 1)}|_{\gf(2)}) = 3$ if $s$ is odd, which follows from (\ref{eqn-19oct171}). 
This is why we are not interested in this code for the case $s$ being odd.

\begin{theorem}
Let $s \geq 4$ be an even integer. Then the code $(\C_{(2^s, 2^s+1, 3,1)}|_{\gf(2)})^\perp$ has parameters 
$[2^s+1, 2s, 2^{s-1}-2^{s/2}+2]$. 
\end{theorem} 

\begin{proof}
$(\C_{(2^s, 2^s+1, 3,1)}|_{\gf(2)})^\perp$ is the dual of the Zetterberg code whose parameters are from 
\cite[Theorem 6.6]{LW90}. 
\end{proof}

When $s=2$, $(\C_{(2^s, 2^s+1, 3,1)}|_{\gf(2)})^\perp$ has parameters $[5, 4, 2]$, and is MDS. 
When $s=4$, the code has parameters $[17, 8, 6]$, and is distance-optimal \cite{Grassl}.  
When $s=4$, the code has parameters $[65, 12, 26]$, and has the best known 
parameters  \cite{Grassl} and is an optimal cyclic code \cite[Appendix A]{Dingbook15}. Thus, the code  
$(\C_{(2^s, 2^s+1, 3,1)}|_{\gf(2)})^\perp$ is very interesting.

\begin{theorem}
Let  $s \geq 2$. Then the code $\C_{(3^s, 3^s+1, 3, 1)}|_{\gf(3)}$  
has parameters $[3^s+1, 3^s+1-4s, d \geq 4]$. 
\end{theorem} 

\begin{proof}
Let $q=3^s$ and $n=q+1=3^s+1$. We follow the notation in the proofs of Theorems \ref{thm-SQScode0} and \ref{thm-SQScode}.  
Note that the $3$-cyclotomic coset $C_1$ modulo $n$ is given by 
$$ 
C_1=\{1,3, \ldots, 3^{s-1}, -1, -3, \ldots, -3^{s-1}\} \bmod n. 
$$ 
Similarly, the $3$-cyclotomic coset $C_2$ modulo $n$ is given by $C_2=2C_1 \bmod n$. It is easily verified 
that $C_1 \cap C_2 = \emptyset$ and $|C_1|=|C_2|=2s$.

By definition, the minimal polynomial $\M_{\beta^j}(x)$ of $\beta^j$ over $\gf(3)$ is given by 
$$ 
\M_{\beta^j}(x)=\sum_{i \in C_j} (x - \beta^i) 
$$
for $j \in \{1,2\}$. 
By the definition of BCH codes, $\C_{(3, 3^s+1, 3,1)}$ has generator polynomial $\M_{\beta}(x)\M_{\beta^2}(x)$. 
Therefore, the dimension of the code $\C_{(3, 3^s+1, 3,1)}$ is given by 
$$ 
\dim(\C_{(3, 3^s+1, 3,1)})=3^s+1-4s. 
$$  
By (\ref{eqn-1910153}), $\C_{(3^s, 3^s+1, 3, 1)}|_{\gf(3)}$ has the same dimension and generator polynomial as 
$\C_{(3, 3^s+1, 3,1)}$. 

By Theorem \ref{thm-SQScode}, $\C_{(3^s, 3^s+1, 3,1)}$ has minimum distance $4$. It then follows from the definition 
of subfield subcodes that the minimum distance $d(\C_{(3^s, 3^s+1, 3, 1)}|_{\gf(3)}) \geq 4$. This completes the proof. 
\end{proof} 

We inform the reader that the minimum distance of $\C_{(3^s, 3^s+1, 3, 1)}|_{\gf(3)})$ is indeed $4$ when $s=3$.  
We have the following examples of the code $\C_{(3^s, 3^s+1, 3, 1)}|_{\gf(3)}$: 
\begin{eqnarray*}
\begin{array}{ccc}
s  & \C_{(q, q+1, 3,1)}|_{\gf(3)} & (\C_{(q, q+1, 3,1)}|_{\gf(3)})^\perp \\
2 & [10, 2, 5] & [10, 8, 2] \\
3 & [28, 16, 4] & [28, 12, 8] \\
4 & [82, 66, 6] & [82, 16, 36]  
\end{array}
\end{eqnarray*}
$\C_{(q, q+1, 3,1)}|_{\gf(3)}$ and $(\C_{(q, q+1, 3,1)}|_{\gf(3)})^\perp$  both are distance-optimal 
cyclic codes when $s=2$ and $s=3$ according to \cite[Appendix A]{Dingbook15}.   
The distance optimality of these subfield subcodes make the original codes $\C_{(q, q+1, 3,1)}$ 
and $\C_{(q, q+1, 3,1)}^\perp$ very interesting. 

It would be worthy to settle the minimum distances of  $\C_{(q, q+1, 3,1)}|_{\gf(3)}$ and 
$\C_{(q, q+1, 3,1)}|_{\gf(3)})^\perp$.  The reader is invited to attack this open problem.

\section{Summary and concluding remarks}

This paper settled a 70-year-old open problem by presenting an infinite family of near MDS codes over 
$\gf(3^s)$ holding an infinite family of $3$-designs and an infinite family of near MDS codes over $\gf(2^{2s})$ 
holding an infinite family of $2$-designs. Hence, these codes have nice applications in combinatorics. 
The two families of near MDS codes are very interesting in coding theory, as their 
ternary and binary subfield subcodes are dimension-optimal or distance-optimal cyclic codes. 
It would be nice if the automorphism groups of the linear codes could be determined.    

An interesting open problem is whether there exists an infinite family of linear codes holding the spherical design 
$S(3, 1+q, 1+q^m)$ for arbitrary prime power $q$ and $m \geq 3$. This problem was settled only for the special 
case $q=3$ in this paper.  Another open problem is whether there exists an infinite family of near MDS codes holding an 
infinite family of $4$-designs. Quasi-cyclic NMDS codes may be such codes \cite{TD13}.



\end{document}